\documentclass[11pt]{article}

\usepackage[preprint]{acl}

\usepackage{times}
\usepackage{latexsym}

\usepackage[T1]{fontenc}

\usepackage[utf8]{inputenc}

\usepackage{microtype}

\usepackage{inconsolata}

\usepackage{hyperref}
\usepackage{url}
\usepackage{booktabs}
\usepackage{amsmath,amssymb}
\usepackage{graphicx}
\graphicspath{{../}}
\usepackage{subcaption}
\usepackage[ruled,vlined]{algorithm2e}
\usepackage{bbm}
\usepackage{comment}
\usepackage{multirow}

\usepackage{placeins}

\usepackage{amsthm}
\newtheorem{proposition}{Proposition}
\theoremstyle{remark}

\newcommand{\hypothesisbox}[1]{%
\par\smallskip
\noindent\begingroup
\setlength{\fboxsep}{4pt}%
\begin{minipage}{\columnwidth}
\colorbox{black}{%
\parbox{\dimexpr\linewidth-2\fboxsep\relax}{%
\color{white}\bfseries\small Analysis Hypotheses}}%
\par\nointerlineskip
\fcolorbox{black}{gray!6}{%
\parbox{\dimexpr\linewidth-2\fboxsep-2\fboxrule\relax}{\small #1}}%
\end{minipage}%
\endgroup
\par\smallskip}

\title{From Positionwise Confidence to Prefix Scheduling: Verifier Skipping in Speculative Decoding}

\author{
  Haoxuan Luo, Jameson Sandler, Ferdinando Fioretto \\
  University of Virginia \\
  \texttt{\{ayr7tb,jmz4ds,fioretto\}@virginia.edu}
}

\begin{document}
\maketitle
\begin{abstract}
Speculative decoding is a leading technique to reduce the cost of autoregressive generation by using a small drafter to propose several tokens, which are then verified in parallel by a larger target model. Speculative diffusion decoding (SDD) further removes sequential drafting by generating every position in a draft block in parallel with a discrete diffusion model.
However, SDD still invokes the target on every block, leaving verification as a potential bottleneck. This paper recognizes that this creates a new control handle: whether to invoke the verifier at all. Thus, we study \emph{verifier skipping}, a lossy policy that commits a selected draft prefix directly, and ask which confidence signal should schedule it. Interestingly, our study finds that better token predictors need not yield better schedulers: skips require contiguous high-confidence prefixes, while short skips can induce additional drafting rounds. To study this mismatch, we compare raw confidence with learned marginal and conditional survival scores under the same policy, using Strict SDD, lenience, and top-$k$ acceptance as baselines. On HumanEval with DiffuCoder-7B-Instruct and Qwen3-32B, all three confidence signals save $9.6\%$ to $13.5\%$ of verifier calls at the same observed pass@1 as Strict SDD. Surprisingly, raw confidence saves the most; marginal survival has higher positionwise AUROC than raw confidence at most positions, yet neither learned signal dominates online. 
Our analysis shows that verifier skipping is a useful new lossy axis and, surprisingly, its key challenge is prefix scheduling rather than token prediction alone.
\end{abstract}

\section{Introduction}
\label{sec:introduction}

\begin{figure}[t]
\centering
\includegraphics[width=\columnwidth]{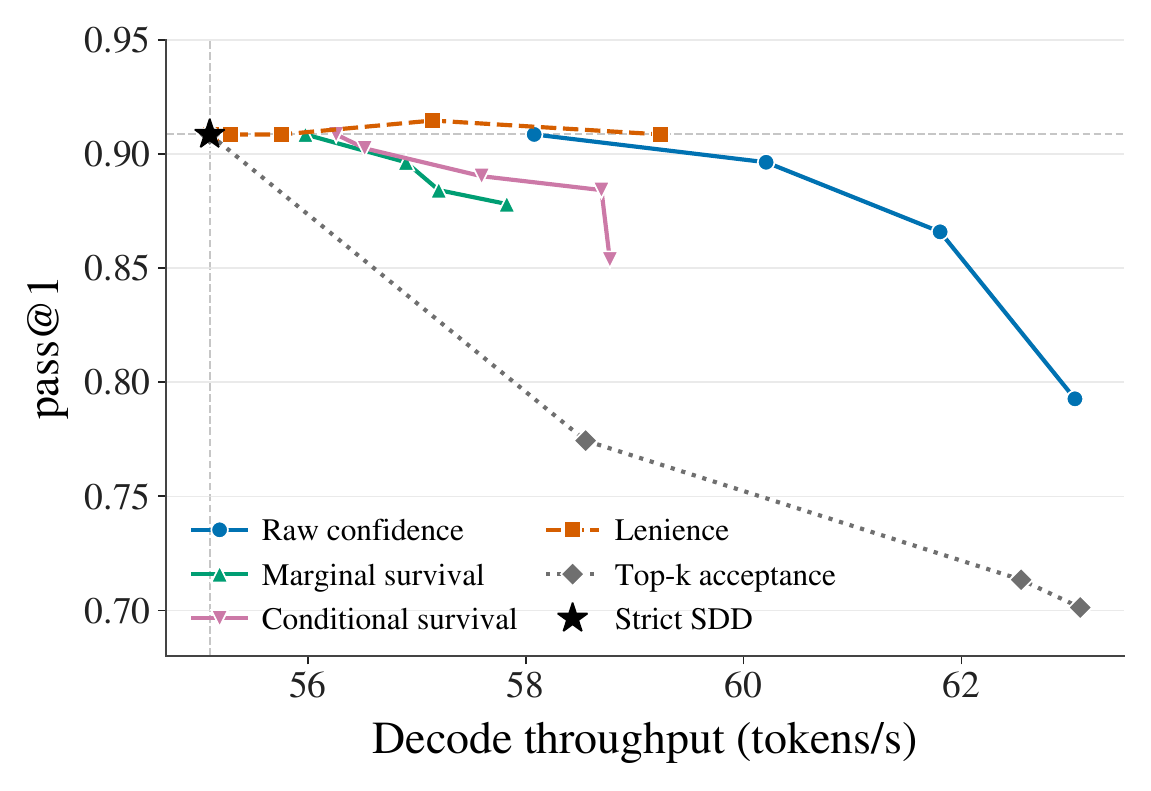}
\caption{
HumanEval pass@1 versus decode throughput with DiffuCoder-7B-Instruct and Qwen3-32B, using draft length $\gamma=32$. Raw confidence, marginal survival, and conditional survival use the same verifier-skipping policy; Strict SDD, lenience, and top-$k$ acceptance are shown as baselines. Curves are obtained by sweeping the fixed parameter grid for each method.
}
\label{fig:main_pareto}
\end{figure}

Speculative decoding reduces generation cost by pairing a small drafter, which proposes several tokens, with a larger target model that verifies them in parallel \cite{leviathan2023fast,chen2023accelerating}. Speculative diffusion decoding (SDD) extends this separation by generating every position in a draft block in parallel with a discrete diffusion model, thereby removing the sequential dependency within the drafting stage \cite{christopher2025speculative}. However, this additional parallelism leaves target-side verification unchanged: strict decoding still invokes the target model for every drafted block, even when the target is substantially larger and its forward pass dominates the latency of a decoding round \cite{kumar2025hispec,zhong2025speeding}. This shift in the computational bottleneck motivates \emph{verifier skipping}, a distinct lossy axis that controls whether the target should be invoked for the current block, rather than changing only how tokens are drafted or accepted. Skipping a call saves verification work by committing a drafted prefix directly, but it may also commit a token that strict verification would reject; consequently, the central challenge is to determine when a sufficiently reliable prefix is worth that risk.

Confidence provides a natural basis for this decision because the scheduler must act before invoking the target and can therefore rely only on information available from the drafter. Selective prediction and model routing use confidence for a related purpose, that is deciding whether a model should return an output or defer to a more expensive system \cite{pmlr-v97-geifman19a,gu-hopkins-2023-evaluation,huang2025confidence,lee2026confidence}. This connection motivates evaluating candidate signals offline through positionwise binary cross entropy (BCE) and area under the receiver operating characteristic curve (AUROC). These metrics characterize local predictive quality, but verifier skipping is not a collection of independent token decisions: the scheduler can commit only a contiguous prefix, so high scores that occur at isolated positions may provide no usable skip. Moreover, the online value of a feasible prefix depends on its length, since a short skip may avoid one verifier call while requiring additional diffusion rounds to complete the same output. 

The paper thus poses two hypotheses about the relationship between offline prediction and online scheduling:
\hypothesisbox{
\textbf{H1. Offline prediction $\not\Rightarrow$ online scheduling.}
Better positionwise BCE or AUROC need not produce a better online frontier because skipping requires contiguous prefixes.

\smallskip
\textbf{H2. Fewer calls $\not\Rightarrow$ higher throughput.}
Short skips can add drafting rounds, so fewer verifier calls need not yield higher throughput.}

To test these hypotheses, we hold the prefix policy fixed and vary only the confidence signal: raw drafter confidence, marginal survival, or conditional survival. A shadow verifier audits every skipped round without affecting generation, allowing us to compare task quality, strict agreement, verifier calls, and throughput on HumanEval with DiffuCoder-7B-Instruct and Qwen3-32B \cite{gong2025diffucoder,yang2025qwen3technicalreport,chen2021evaluatinglargelanguagemodels}. 

\noindent\textbf{Contributions.} This paper makes three contributions. First, it introduces verifier skipping as a new lossy axis for SDD and formulates it as a prefix-scheduling problem. Second, it combines a proposition, a permutation diagnostic, and a minimum-skip-length analysis to explain why better positionwise metrics or fewer verifier calls need not improve online scheduling. Finally, it evaluates raw confidence and two learned survival signals against Strict SDD, lenience, and top-$k$ acceptance \cite{leviathan2023fast}. As Figure~\ref{fig:main_pareto} shows, at the same observed pass@1 as Strict SDD, raw confidence, conditional survival, and marginal survival reach $58.1$, $56.3$, and $56.0$ tokens/s, respectively, compared with $55.1$ tokens/s for Strict SDD. At these points, they reduce verifier calls per generated token by $13.4\%$, $12.5\%$, and $9.6\%$, respectively. The learned signals do not consistently improve the online frontier over raw confidence. These contributions show that effective verifier skipping depends on feasible prefixes and decoding dynamics, not token prediction alone.

\section{Related Work}
\label{sec:related_work}

\noindent\textbf{Speculative diffusion decoding.}
SDD pairs a parallel discrete-diffusion drafter with an autoregressive verifier \cite{christopher2025speculative}. Subsequent methods retain this verifier while improving the drafter: SpecDiff-2 scales drafter--verifier alignment \cite{sandler2025specdiff}, DFlash incorporates target-model information into block-diffusion drafting \cite{chen2026dflash}, and DiffuSpec combines causal path search with adaptive proposal lengths \cite{li2026diffuspec}. A second line of work changes the proposal boundary. SpecDec++ adapts the candidate length and AdaEDL stops drafting using a lower bound on token acceptance \cite{huang2024specdec++,agrawal2024adaedl}, while Speculative Verification and DSpark limit how much of the proposal the target checks \cite{kim-etal-2026-speculative,cheng2026dspark}. DSpark is especially close because prefix survival determines the verification length, but the target is still invoked. Lenience instead relaxes target acceptance \cite{leviathan2023fast}, whereas SpecPV verifies only part of the proposal \cite{tan2025specpv}. These methods change drafting, acceptance, or the work performed after the target is called. Verifier skipping decides whether to call the target for the current block. If the call is skipped, the policy commits a draft prefix without seeing the target output for that block.

\noindent\textbf{Routing target computation.}
Verifier skipping is most closely related to systems that invoke a larger model conditionally. BiLD uses drafter confidence to trigger an autoregressive large decoder \cite{kim2023speculative}, while SPRINTER uses predicted rejection to schedule approximate verification \cite{zhong2025speeding}. S2D2 brings conditional verification to diffusion language models but uses one block-diffusion model for both drafting and routed autoregressive verification \cite{han2026s2d2}. In device--edge inference, U-HLM uses local uncertainty to invoke a remote language model \cite{oh2025uncertaintyawarehybridinferenceondevice}, whereas CoVSpec uses a probability margin to coordinate device--edge vision-language inference \cite{jia2026covspecefficientdeviceedgecoinference}. BiLD, SPRINTER, U-HLM, and CoVSpec route target computation around autoregressive generation, while S2D2 uses one block diffusion model for both drafting and routed verification. Our setting uses a separate parallel diffusion drafter and autoregressive verifier. For each draft block, the scheduler decides whether to call the verifier and, if not, which contiguous prefix to commit using only drafter-side information.


\noindent\textbf{Confidence-based deferral.}
Selective prediction studies when a model should abstain and how such decisions should be evaluated in NLP \cite{pmlr-v97-geifman19a,gu-hopkins-2023-evaluation}. Confidence has also been used to decide whether to answer, which model should handle a query, or whether to continue a cascade \cite{machcha-etal-2025-large,su2026cp,soiffer-etal-2025-semantic,shen-etal-2025-sater,ong2025routellmlearningroutellms,chen2023frugalgptuselargelanguage}. Uncertainty and calibration estimates can vary across tasks and estimators, and entropy alone may be insufficient for safe selective prediction \cite{tao2025revisitinguncertaintyestimationcalibration,phillips2026entropyinsufficientsafeselective}. Verifier skipping instead repeatedly decides whether to commit a contiguous prefix, so its behavior also depends on prefix length and the drafting rounds that follow.

\section{Problem Setup}
\label{sec:problem_setup}

\subsection{Strict Decoding and Verifier Skipping}

Speculative decoding accelerates autoregressive generation by pairing a drafter with a larger verifier \cite{leviathan2023fast}. Given the current prefix $x$, the drafter proposes candidate tokens, which the verifier processes under a rule that preserves the target model's decoding distribution. We focus on greedy decoding in this work and do not evaluate sampling-based decoding.

SDD retains this separation but generates the draft block $\hat{x}_{1:\gamma}=(\hat{x}_1,\ldots,\hat{x}_\gamma)$ in parallel \cite{christopher2025speculative}. The autoregressive verifier accepts a prefix of the block, whose length we denote by $L\in\{0,\ldots,\gamma\}$. Strict SDD invokes the verifier for every block before committing the resulting output under the standard rule. Because the verifier is substantially larger than the drafter in our setting, this mandatory call accounts for much of each round's cost.

Verifier skipping turns this mandatory call into a scheduling decision. Let $K\in\{0,\ldots,\gamma\}$ denote the committed skip length: $K=0$ executes the standard strict round, whereas $K>0$ commits $\hat{x}_{1:K}$ without calling the verifier. Comparing $K$ with $L$ identifies when a skip is lossy, since $K>L$ means that the committed prefix contains at least one token that strict verification would reject. This comparison also provides the basis for measuring agreement.

\subsection{Agreement with Strict Verification}
\label{subsec:shadow_audit}

The comparison between $K$ and $L$ motivates two complementary agreement measures. Let $\mathcal A$ be the set of rounds in which a relaxed rule commits a nonempty draft prefix. For each $b\in\mathcal A$, let $K_b$ denote the committed prefix length and $L_b$ the length that strict verification would accept from the same draft block.

Because $L_b$ is unobserved on skipped rounds, a shadow verifier recovers it from the same prefix and draft block. Its output is used only for agreement: it does not affect decoding, and the call is excluded from latency and verifier call counts. Lenience and top-$k$ acceptance already invoke the target, so their $L_b$ is obtained from that call. In each case, the comparison remains local to the current block rather than a separately generated Strict SDD sequence.

The first metric, strict token agreement, measures the fraction of committed tokens that strict verification would accept:
\begin{equation}
\mathrm{StrictToken}
=
\frac{
\sum_{b\in\mathcal A}\min(K_b,L_b)
}{
\sum_{b\in\mathcal A}K_b
}.
\label{eq:strict_token_agreement}
\end{equation}
Because the metric is token weighted, it gives partial credit when strict verification would reject only a suffix of the committed prefix. To distinguish such partially valid commits from prefixes accepted in full, we also report full prefix acceptance:
\begin{equation}
\mathrm{FullPrefix}
=
\frac{1}{|\mathcal A|}
\sum_{b\in\mathcal A}
\mathbbm{1}\{K_b\leq L_b\}.
\label{eq:full_prefix_acceptance}
\end{equation}
It requires every committed token to pass strict verification; both metrics equal one for Strict SDD.

\section{Confidence-Guided Verifier Skipping}
\label{sec:skip_policy}

Figure~\ref{fig:method_overview} shows the overall decoding loop and how the position scores are used to select a candidate prefix.

\begin{figure*}[t]
\centering
\includegraphics[width=\textwidth]{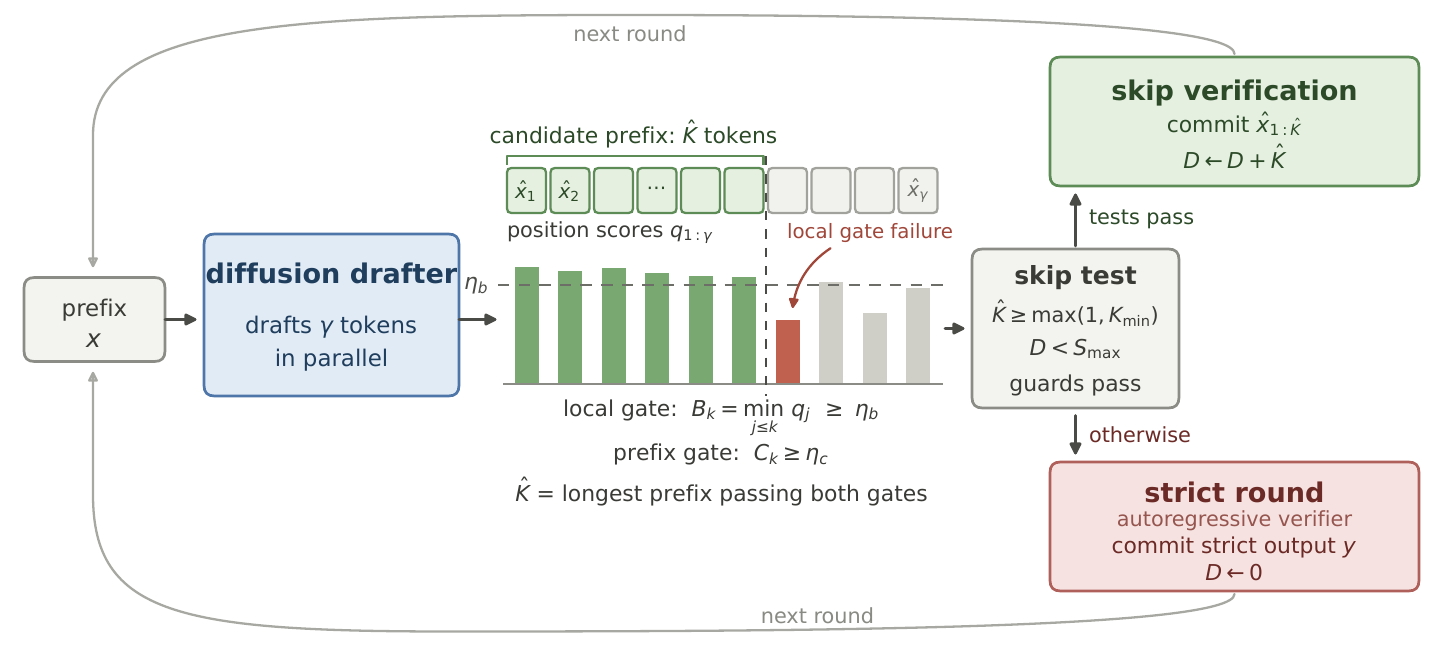}
\caption{
Confidence-guided verifier skipping. The outer path shows the decoding loop: the drafter generates a block, then the policy either commits the selected prefix directly without verification or executes a strict round. The center shows how position scores are used to select the longest prefix passing the local and prefix gates. The skip test checks the minimum skip length, staleness, and remaining guards.
}
\label{fig:method_overview}
\end{figure*}

\subsection{Skip Policy}
\label{subsec:skip_policy}

At each round, the diffusion drafter proposes $\hat{x}_{1:\gamma}$ and a confidence estimator assigns each position a score $q_j\in[0,1]$. Section~\ref{subsec:confidence_signals} defines three interpretations of this score within the same policy.

For a prefix of length $k$, the local score $B_k=\min_{1\leq j\leq k}q_j$ rejects any prefix containing a low-confidence position, while the signal-specific score $C_k=\mathcal C_k(q_{1:k})$ summarizes the prefix as a whole.

Given thresholds $\eta_b$ and $\eta_c$, the policy selects the longest prefix that satisfies both gates:
\begin{equation}
\widehat K
=
\max\!\left\{
k\!\in\!\{1,\ldots,\gamma\}:
B_k\geq\eta_b,\;
C_k\geq\eta_c
\right\},
\label{eq:candidate_skip_length}
\end{equation}
where $\max\varnothing=0$. For example, in Figure~\ref{fig:method_overview}, the first six positions form the longest prefix passing both gates, so $\widehat K=6$. If $K_{\min}=6$, $D<S_{\max}$, and the remaining guards pass, the policy skips verification and commits these six tokens directly; otherwise, it executes a strict round.

Confidence alone does not determine whether an eligible prefix should be committed, because the policy must also control fragmentation and stale target feedback. Let $D$ count the unverified tokens committed since the most recent verifier call, and let $S_{\max}$ denote the staleness threshold. If $D\geq S_{\max}$, the policy forces a strict round; otherwise, it skips only when $\widehat K>0$, $\widehat K\geq K_{\min}$, and all remaining guards pass. The minimum length $K_{\min}$ prevents short skips from fragmenting generation, whereas the staleness rule restores target feedback after enough unverified tokens. The remaining guards cover end-of-sequence and other special tokens, repetition, invalid scores, and shortened final blocks.

When all conditions hold, the policy sets $K=\widehat K$, commits $\hat{x}_{1:K}$, and updates $D\leftarrow D+K$; otherwise, it sets $K=0$ and executes a strict round. Every verifier call resets $D$ to zero. Because staleness is checked before a skip, $D$ may exceed $S_{\max}$ by at most $\gamma-1$, in which case the next round is necessarily strict.

\subsection{Confidence Signals}
\label{subsec:confidence_signals}

To isolate the confidence representation from the scheduler, we evaluate three signals within the same policy. They differ only in the position score $q_j$ and prefix rule $\mathcal C_k$. The two learned variants also share the same per-position multilayer perceptron (MLP) and pre-verification inputs, differing only in their training targets and examples. Table~\ref{tab:confidence_signals} summarizes the three signals.

\begin{table}[t]
\centering
\caption{Position and prefix scores for the three confidence signals.}
\label{tab:confidence_signals}
\small
\setlength{\tabcolsep}{4pt}
\begin{tabular}{@{}lll@{}}
\toprule
\textbf{Signal} & \textbf{Position score $q_j$} & \textbf{Prefix score $C_k$} \\
\midrule
Raw confidence
& $p_{\mathrm{draft}}(\hat{x}_j)$
& $\left(\prod_{j=1}^k q_j\right)^{1/k}$ \\

Marginal survival
& $\hat a_j$
& $\min_{j\leq k}\hat a_j$ \\

Conditional survival
& $\hat s_j$
& $\prod_{j=1}^k\hat s_j$ \\
\bottomrule
\end{tabular}
\end{table}

The first signal is raw confidence, $q_j^{\mathrm{raw}}=p_{\mathrm{draft}}(\hat{x}_j)$, which under greedy drafting is the largest output probability at position $j$. Because it does not estimate verifier acceptance, we use it as a scheduling heuristic and define $C_k^{\mathrm{raw}}=(\prod_{j=1}^k q_j^{\mathrm{raw}})^{1/k}$. This geometric mean summarizes the prefix, while the local gate detects any low-confidence position.

The second signal, marginal survival, estimates $a_j=\Pr(L\geq j\mid\mathcal F)$ from the pre-verification drafter information $\mathcal F$. It is trained on all blocks using labels $y_j=\mathbbm{1}\{L\geq j\}$. Although the true sequence $a_j$ is nonincreasing, the predictions $\hat a_j$ need not be; we therefore define $\bar a_k=\min_{1\leq j\leq k}\hat a_j$ and set $q_j^{\mathrm{marg}}=\hat a_j$ and $C_k^{\mathrm{marg}}=\bar a_k$. Consequently, $B_k^{\mathrm{marg}}=C_k^{\mathrm{marg}}=\bar a_k$, so both gates reduce to $\eta_m=\max\{\eta_b,\eta_c\}$.

The third signal, conditional survival, estimates $s_j=\Pr(L\geq j\mid L\geq j-1,\mathcal F)$, with $L\geq0$ by convention. Consistent with this at-risk interpretation, positions $j\leq L$ receive positive labels, $j=L+1$ receives a negative label when $L<\gamma$, and later positions are masked. We set $q_j^{\mathrm{cond}}=\hat s_j$ and $C_k^{\mathrm{cond}}=\prod_{j=1}^k\hat s_j$, computing the product in log space. We parameterize $\eta_c=1-\epsilon$, so $\epsilon$ controls the predicted prefix risk.

Because the signals aggregate position scores differently, they also determine whether the local and prefix gates impose distinct constraints. For any $q_{1:k}\in[0,1]^k$,
\begin{equation}
\prod_{j=1}^k q_j
\leq
\min_{j\leq k}q_j
\leq
\left(\prod_{j=1}^k q_j\right)^{1/k}.
\label{eq:gate_order}
\end{equation}
This ordering explains the role of each gate. For raw confidence, the geometric mean can remain high despite one low position, so the local gate adds a distinct constraint. For marginal survival, both gates equal $\bar a_k$ and are therefore identical. For conditional survival, the prefix gate implies the local gate whenever $\eta_c\geq\eta_b$; only when $\eta_b>\eta_c$ can the local gate impose an additional restriction.

\subsection{From Positionwise Prediction to Prefix Scheduling}
\label{subsec:representation_analysis}

\paragraph{Survival semantics.}

Marginal and conditional survival encode the same prefix event at different levels of factorization. In particular,
\begin{equation}
\begin{aligned}
a_K
&=
\Pr(L\geq K\mid\mathcal F) = \prod_{j=1}^K s_j,
\\
s_j
&=
\frac{a_j}{a_{j-1}}
\quad\text{when }a_{j-1}>0,
\qquad a_0=1.
\end{aligned}
\label{eq:survival_equivalence}
\end{equation}
The factorization follows directly from the chain rule and requires no independence assumption. Exact marginal and conditional probabilities therefore contain the same information about prefix survival, even though separately trained predictors may behave differently because their targets and training populations differ.

This equivalence also clarifies how the scores should be aggregated. Only conditional survival factors should be multiplied, because each marginal $a_j$ already includes survival through all earlier positions. For example, if the conditional failure probability is $h$ at every position, then $a_K=(1-h)^K$, whereas multiplying the marginals yields $\prod_{j=1}^K a_j=(1-h)^{K(K+1)/2}$ and therefore counts earlier survival events repeatedly.

The cumulative survival score further provides a direct risk interpretation. Suppose that $\widehat C_k\leq\Pr(L\geq k\mid\mathcal F)$ for every $k$, and that $\widehat K$ is selected using only information in $\mathcal F$. Then
\begin{equation}
\widehat C_{\widehat K}\geq 1-\epsilon
\quad\Longrightarrow\quad
\Pr(L\geq\widehat K\mid\mathcal F)\geq 1-\epsilon.
\label{eq:risk_interpretation}
\end{equation}
Under these assumptions, every selected prefix has conditional probability at least $1-\epsilon$ of being fully accepted. Our learned scores are not certified lower bounds, so this statement provides an interpretation rather than a formal guarantee; raw confidence has no analogous survival semantics.

\paragraph{Positionwise metrics and prefix events.}

Positionwise predictive quality need not determine scheduling quality. Positionwise metrics score each draft position separately, whereas the policy requires several consecutive leading positions to pass together. Two score processes can have the same score-label distribution at every position but differ in how high scores occur in the same block. If high scores occur together near the beginning of a block, long eligible prefixes are common. If they are spread across different blocks, the same positionwise metrics can produce far fewer usable prefixes. The following proposition shows this formally.

\begin{proposition}[Positionwise metrics do not determine prefix frequency]
\label{prop:prefix_frequency}
Fix $1 \leq K\leq\gamma$, $\rho\in(0,1)$, a threshold $\tau\in(0,1)$, and any distribution of $L$ with $\Pr(L\geq K)>0$. There exist two score processes $q_{1:\gamma}$ and $q'_{1:\gamma}$ taking values in $(0,1)$ such that $(q_j,\mathbbm{1}\{L\geq j\})$ and $(q'_j,\mathbbm{1}\{L\geq j\})$ have the same joint distribution at every position $j$, while
\begin{align*}
\Pr\left(\min_{j\leq K}q_j\geq\tau\right)
&=
\Pr(L\geq K)\rho,
\\
\Pr\left(\min_{j\leq K}q'_j\geq\tau\right)
&=
\Pr(L\geq K)\rho^K.
\end{align*}
\end{proposition}

See Appendix~\ref{app:proof} for the proof and a discussion of the construction.

\paragraph{Runtime geometry.}

Prefix frequency is also insufficient to determine online efficiency because skip length changes the number of required draft blocks. Let $c_D$ and $c_V$ denote the average costs of one block and one verifier call. Decode time is approximately
\begin{equation}
T_{\mathrm{decode}}
\approx
c_D N_{\mathrm{blk}}
+
c_V N_{\mathrm{ver}}
+
T_{\mathrm{policy}},
\label{eq:runtime_decomposition}
\end{equation}
where $N_{\mathrm{blk}}$ and $N_{\mathrm{ver}}$ count draft blocks and verifier calls, while $T_{\mathrm{policy}}$ captures predictor inference and other policy overhead.

To isolate the local break-even condition, we omit policy overhead and compare a strict round that commits $A$ tokens on average at cost $c_D+c_V$ with a skip of length $K$ at cost $c_D$. The skip has higher local throughput exactly when
\begin{equation}
\begin{aligned}
\frac{K}{c_D}
&>
\frac{A}{c_D+c_V},
\\
K
&>
K^\star
:=
A\frac{c_D}{c_D+c_V}.
\end{aligned}
\label{eq:break_even_length}
\end{equation}
Thus, local throughput improves only when $K>K^\star$, motivating $K_{\min}$: shorter skips may save a verifier call but add enough draft blocks to reduce throughput.

\section{Experiments}
\label{sec:experiments}

\subsection{Experimental Setup}
\label{subsec:experimental_setup}

\begin{figure*}[t]
\centering
\includegraphics[width=0.485\textwidth]
{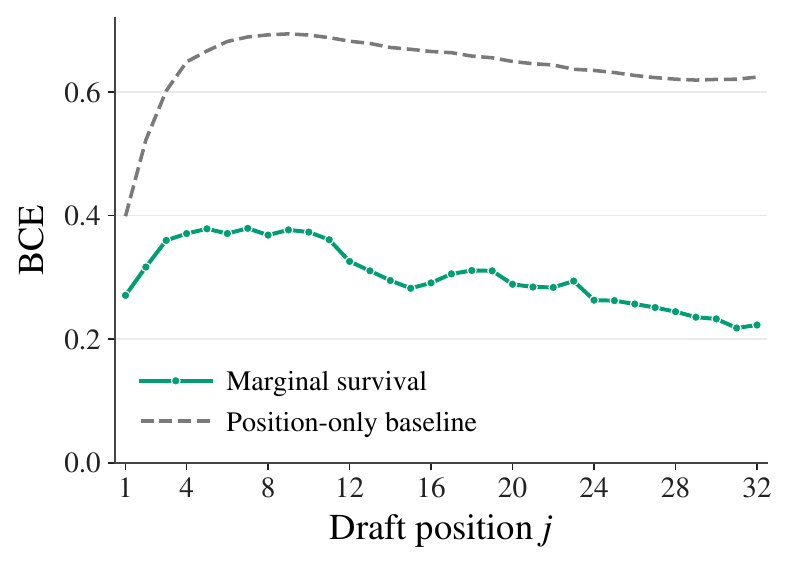}
\hfill
\includegraphics[width=0.485\textwidth]
{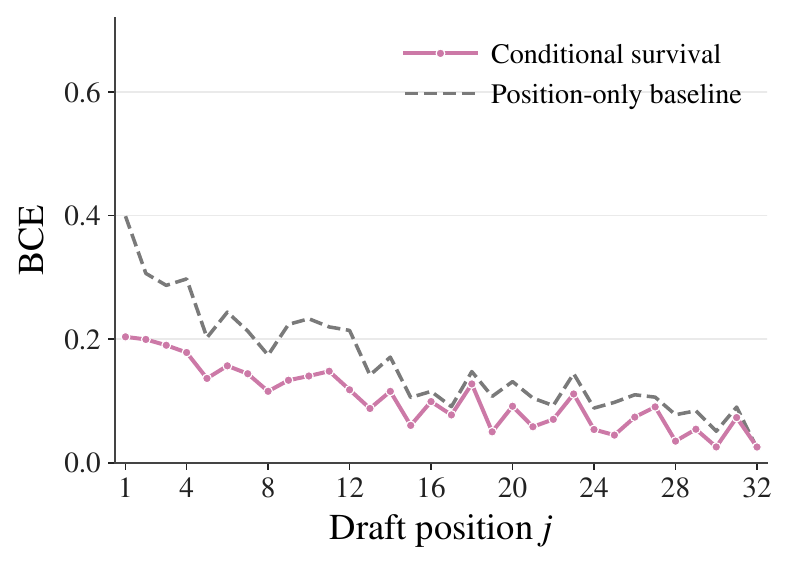}
\caption{
Positionwise BCE for marginal and conditional survival on HumanEval with draft length $\gamma=32$. Dashed curves show position-only baselines estimated from the training folds.
}
\label{fig:offline_bce}
\end{figure*}

We pair DiffuCoder-7B-Instruct with Qwen3-32B as the diffusion drafter and autoregressive verifier, respectively, and disable thinking \cite{gong2025diffucoder,yang2025qwen3technicalreport}. We evaluate all 164 HumanEval problems using pass@1 \cite{chen2021evaluatinglargelanguagemodels}. Unless stated otherwise, decoding is greedy, uses two diffusion steps, and generates at most 512 tokens. A preliminary Strict SDD throughput comparison determined $\gamma=32$ before any skipping-policy evaluation; the main experiments then use $K_{\min}=6$ and $S_{\max}=64$. All main runs use two NVIDIA A100 80 GB GPUs.

To separate predictor training from policy evaluation, we partition HumanEval into five folds. Each run uses three folds for training, one for validation, and one for testing, with the roles rotated so every prompt appears exactly once in a test fold. Raw confidence requires no training but is evaluated on the same test folds, and every frontier uses a fixed threshold grid.

This protocol supports complementary offline and online evaluations. Offline prediction is measured by per-position BCE against position-only baselines, with AUROC as a ranking metric. Online efficiency is measured by verifier calls per generated token, normalized by Strict SDD, and by generated tokens per second; strict token agreement follows Section~\ref{subsec:shadow_audit}. Runtime includes the drafter, verifier, predictor, and policy but excludes the shadow audit. All reported results pool the five test folds.

\subsection{Offline Prediction and Online Frontiers}
\label{subsec:offline_online}

\paragraph{H1: Offline prediction $\not\Rightarrow$ online scheduling.}

To test whether offline improvements predict online scheduling, we first measure how well the learned models estimate strict prefix survival at individual draft positions. Marginal survival predicts $\mathbbm{1}\{L\geq j\}$ over all blocks, whereas conditional survival predicts the same event only among blocks with $L\geq j-1$. Figure~\ref{fig:offline_bce} compares both models with the empirical positive rate at each position, estimated from the training folds.

Both learned predictors improve substantially over their position-only baselines: average BCE, weighted by the number of test examples at each position, decreases from $0.641$ to $0.306$ for marginal survival and from $0.195$ to $0.123$ for conditional survival. Conditional estimates become noisier at later positions because fewer blocks remain at risk and provide training examples. Appendix~\ref{app:offline_metrics} complements this comparison with per-position AUROC for the learned predictors and raw confidence, which has slightly higher AUROC at some early positions.

\paragraph{Online scheduling.}

Having established an offline predictive advantage, we next ask whether it transfers online by sweeping the three signals' thresholds while holding all other policy components fixed. Figure~\ref{fig:main_pareto} shows pass@1 versus throughput, while Figure~\ref{fig:online_frontiers} shows pass@1 and strict token agreement versus relative verifier calls. Strict SDD provides the common quality and cost reference, with strict token agreement equal to one.

\begin{figure*}[t]
\centering
\centering
\includegraphics[width=0.485\linewidth]{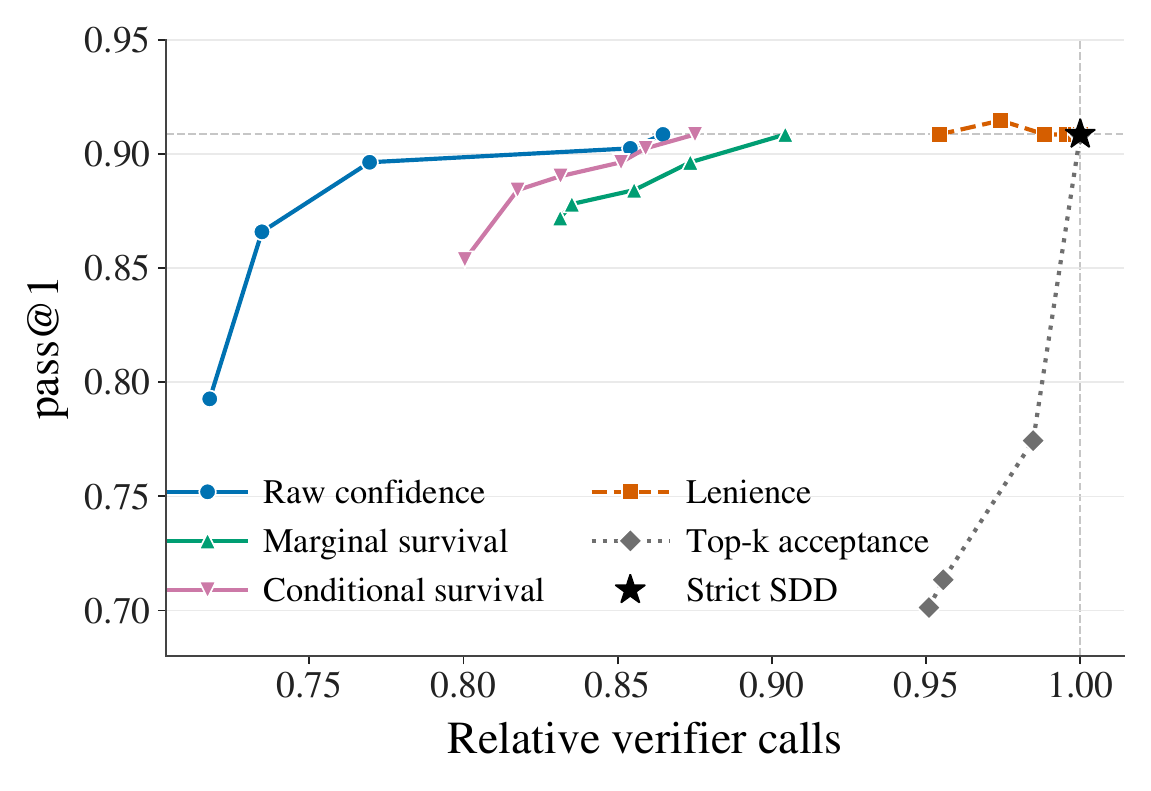}
\hfill
\centering
\includegraphics[width=0.485\linewidth]{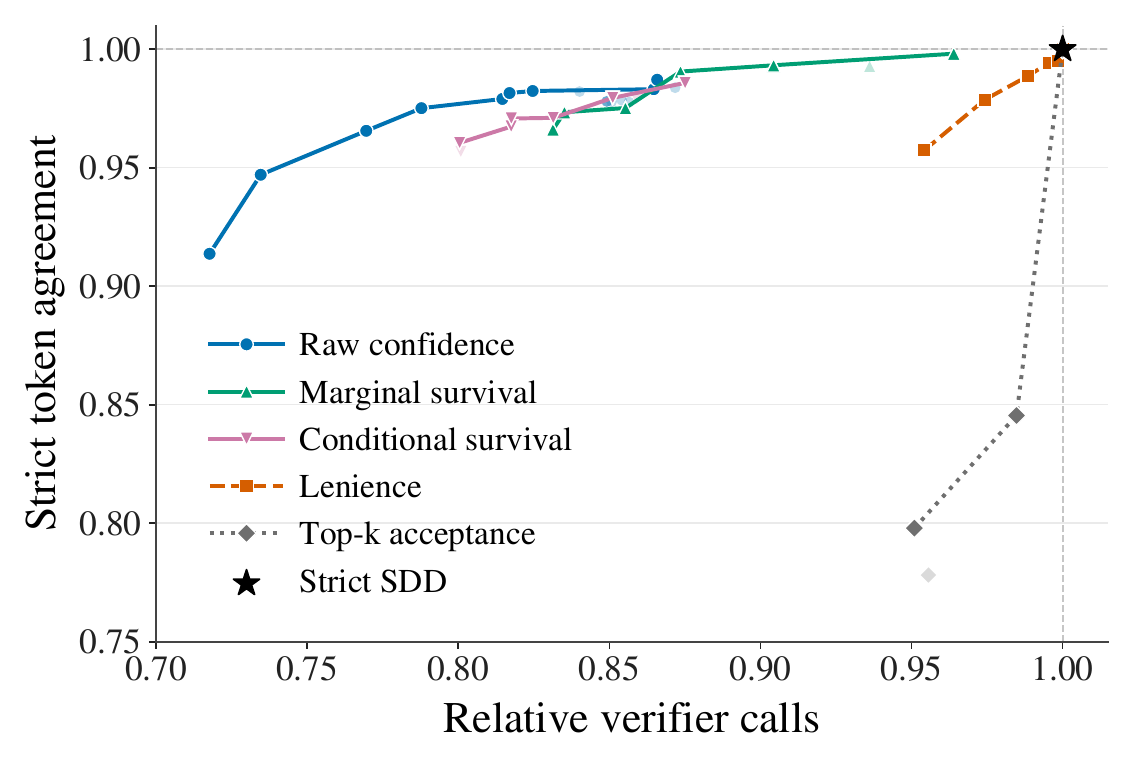}
\caption{
HumanEval pass@1 and strict token agreement versus relative verifier calls with draft length $\gamma=32$. Curves connect the plotted settings for each method; faded markers show additional settings.
}
\label{fig:online_frontiers}
\end{figure*}

Strict SDD obtains $0.9085$ pass@1. At observed points with the same aggregate pass@1, raw confidence, conditional survival, and marginal survival reduce verifier calls by $13.4\%$, $12.5\%$, and $9.6\%$, respectively. Their full prefix acceptance is $0.9814$, $0.9712$, and $0.9776$, respectively, showing that each policy saves calls while remaining close to strict verification. 
Using 10,000 paired bootstrap resamples of the HumanEval prompts, the 95\% intervals for the difference in pass@1 relative to Strict SDD are $[0,0]$, $[-0.018,0.018]$, and $[0,0]$ for raw confidence, conditional survival, and marginal survival, respectively. The corresponding intervals for throughput gain are $[1.43,4.92]$, $[-1.24,3.54]$, and $[-0.77,3.43]$ tokens/s; these intervals reflect variation across prompts rather than repeated timing runs.

The lossy baselines provide additional context for these gains. Lenience still verifies every block, yet its calls per generated token are $4.6\%$ below Strict SDD at the same pass@1, while top-$k$ acceptance provides little call reduction and falls to $0.7744$ pass@1 at $k=2$. A periodic policy that attempts a fixed-length skip once every $P$ rounds without confidence is also dominated by raw confidence in both pass@1 and relative verifier calls at every tested setting; Appendix~\ref{app:periodic} reports these results.

\noindent\textbf{H1 finding and significance.} Despite their offline gains, the learned signals do not dominate online. Raw confidence reaches the lowest call rates in the aggressive region and remains on the frontier near Strict SDD quality; conditional survival uses fewer calls than marginal survival at several high-quality points, whereas marginal survival achieves higher strict token agreement at conservative settings. Thus, no signal dominates across task quality, verifier calls, and strict agreement. This mismatch supports H1: better positionwise prediction is not sufficient for a better scheduling frontier. Importantly, the result shows why BCE or AUROC alone cannot select an online scheduler; it establishes a lack of observed dominance by the learned scores, not universal superiority of raw confidence.

\noindent\textbf{Why the online metrics must be separated.} At a raw-confidence point matching Strict SDD pass@1, throughput rises from $55.1$ to $58.1$ tokens per second as relative verifier calls fall to $0.866$. A more aggressive raw-confidence point reaches $59.8$ tokens/s with $0.8963$ pass@1 and $0.975$ strict token agreement. Lenience reaches $59.2$ tokens/s at the same aggregate pass@1 as Strict SDD, with $0.954$ relative calls and $0.957$ agreement. The $58.1$ tokens/s raw-confidence point has $0.987$ agreement and matches Strict SDD's pass/fail outcome on every problem, whereas lenience gains two passes and loses two. Matching aggregate pass@1 can conceal different call, throughput, and agreement behavior, motivating the H2 analysis below.

\begin{figure*}[t]
\centering
\begin{subfigure}[t]{0.48\textwidth}
\centering
\includegraphics[width=\linewidth]{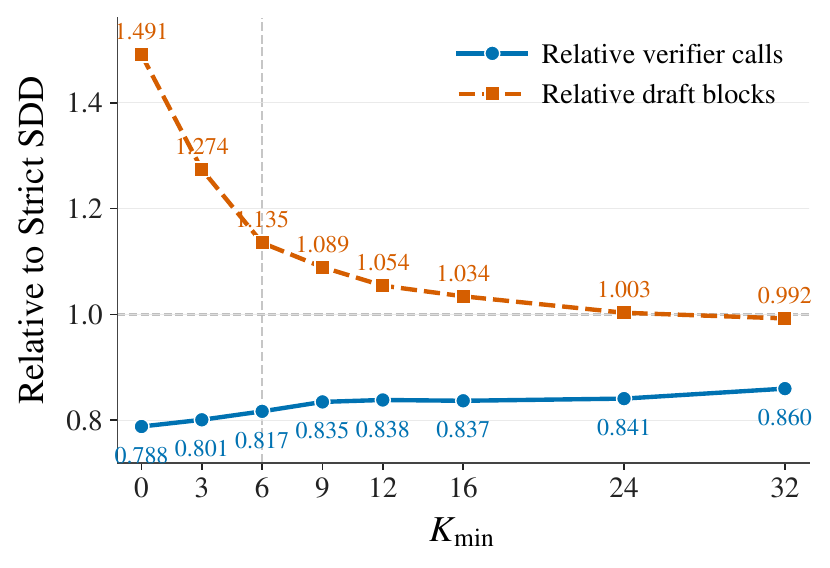}
\end{subfigure}
\hfill
\begin{subfigure}[t]{0.48\textwidth}
\centering
\includegraphics[width=\linewidth]{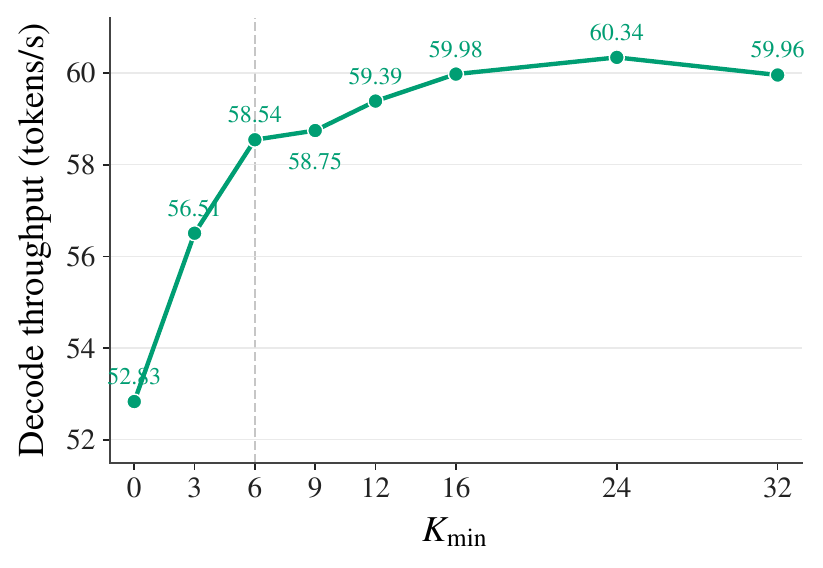}
\end{subfigure}
\caption{
Effect of $K_{\min}$ on verifier calls, draft fragmentation, and throughput for raw confidence. The remaining policy thresholds are fixed, and all points use the same timing setup.
}
\label{fig:kmin_sweep}
\end{figure*}

\subsection{Prefix Geometry and Runtime}

\paragraph{Why H1 arises: dependence within a block.}
The mismatch has a structural source: offline metrics evaluate positions separately, whereas the scheduler can skip only a contiguous prefix. Scheduling therefore depends on whether favorable scores occur together near the beginning of the same block. Proposition~\ref{prop:prefix_frequency} formalizes this gap by showing that positionwise metrics cannot determine how often eligible prefixes occur.

To isolate this dependence empirically, we apply an offline permutation diagnostic to the logged test blocks. At each position $j$, we randomly permute recorded score--label pairs across blocks. Raw confidence and marginal survival use all valid blocks with label $\mathbbm{1}\{L\geq j\}$, while conditional survival uses only blocks with $L\geq j-1$. Moving each pair together preserves every position's score--label distribution, and hence every positionwise metric, while breaking the association among positions from the same block.

We repeat the permutation 100 times at fixed settings: $(\eta_b,\eta_c)=(0.90,0.93)$ for raw confidence, $\eta_m=0.90$ for marginal survival, and $(\eta_b,\eta_c)=(0.90,0.90)$ for conditional survival. Each repetition recomputes the longest prefix satisfying the score gates before $K_{\min}$ or staleness is applied; we call this the feasible prefix length. Because the diagnostic uses only logged blocks, it does not regenerate a decoding trajectory.

\begin{table}[t]
\centering
\caption{
Mean feasible prefix length before and after permutation. Per-position BCE and AUROC are unchanged by construction.
}
\label{tab:permutation}
\small
\setlength{\tabcolsep}{7pt}
\begin{tabular}{@{}lcc@{}}
\toprule
Signal & Original & Permuted mean \\
\midrule
Raw confidence       & 8.75 & 1.59 \\
Marginal survival    & 6.86 & 1.00 \\
Conditional survival & 8.76 & 2.11 \\
\bottomrule
\end{tabular}
\end{table}

\noindent\textbf{Mechanistic evidence for H1.} Table~\ref{tab:permutation} shows a large reduction for every signal: original mean lengths range from $6.86$ to $8.76$, whereas permuted means range from $1.00$ to $2.11$. Because all positionwise metrics are unchanged, this collapse isolates information that BCE and AUROC cannot capture---which positions become jointly available as a prefix. The diagnostic does not reproduce an online trajectory, but it explains one mechanism behind the offline-to-online gap. Appendix~\ref{app:gate_ablation} complements this analysis with a fixed-threshold ablation of the local and prefix gates.

\paragraph{H2: Fewer verifier calls $\not\Rightarrow$ higher throughput.}

Equation~\ref{eq:break_even_length} predicts that short skips may reduce verifier calls without improving throughput. Figure~\ref{fig:kmin_sweep} tests this mechanism by varying $K_{\min}$ while holding the remaining thresholds fixed. Setting $K_{\min}=0$ removes the minimum-length requirement, although a skip still requires $\widehat K>0$. Raw confidence then lowers relative verifier calls to $0.788$, but raises relative draft blocks to $1.491$ and reaches only $52.83$ tokens per second. Increasing $K_{\min}$ to $6$ raises relative calls to $0.817$, yet reduces relative blocks to $1.135$ and increases throughput to $58.54$ tokens per second.

The same pattern continues across the sweep: throughput reaches $60.34$ tokens/s at $K_{\min}=24$ and is slightly lower at $K_{\min}=32$, even as relative verifier calls rise from $0.788$ to $0.860$. Observed pass@1 is $0.8963$ for $K_{\min}\in\{0,3,9\}$, $0.8902$ for $K_{\min}=6$, $0.9024$ for $K_{\min}\in\{12,16\}$, and $0.9085$ for $K_{\min}\in\{24,32\}$.

\noindent\textbf{H2 finding and significance.} Minimizing verifier calls does not maximize throughput, supporting H2 and showing why call reduction alone is an incomplete efficiency objective: short skips can replace one verifier call with enough additional drafting work to slow decoding. We therefore use $K_{\min}=6$ in the main experiments, which removes most additional draft blocks while keeping verifier calls $18.3\%$ below Strict SDD. Higher values reduce fragmentation further but yield smaller throughput gains and fewer call savings.

\section{Conclusion}

We introduced verifier skipping as a distinct lossy axis for SDD and evaluated the same prefix policy with raw confidence, marginal survival, and conditional survival. On HumanEval with DiffuCoder-7B-Instruct and Qwen3-32B, the three signals reach the same observed pass@1 as Strict SDD at operating points that reduce verifier calls per generated token by $9.6\%$ to $13.5\%$. Raw confidence gives the largest reduction, and neither learned signal consistently improves the observed online tradeoff over it. 
The evidence supports both analysis hypotheses: positionwise prediction metrics alone do not capture the prefix structure used by the scheduler, and fewer verifier calls do not necessarily produce higher throughput.

\section*{Limitations}
The main experiments use one drafter and verifier pair on HumanEval under greedy decoding. Appendix~\ref{app:qwen25_transfer} reports earlier transfer results on two additional benchmarks, but we have not tested other current model pairs or decoding settings. Each reported HumanEval point uses one completion per prompt, and matched-pass@1 operating points are identified from the observed test frontier. Equality in aggregate pass@1 on 164 problems does not establish statistical quality equivalence; our results show no observed online dominance by the learned signals, not a universal ordering. Throughput is measured using two GPUs with full sequence recomputation and no verifier KV cache, so the throughput results may differ in cached serving systems. Each prompt and setting is timed once, so small differences in tokens/s may be due to timing noise. The bootstrap intervals reported above resample prompts and do not measure timing variation across repeated runs. The learned survival scores are not certified lower bounds and therefore provide no formal guarantee of agreement with strict decoding. The predictors are trained only on Strict SDD trajectories and may encounter distribution shift after earlier verifier skips. We evaluate only one MLP architecture trained with positionwise objectives.

\bibliography{custom}

\appendix

\section{Additional Experimental Details}
\label{app:experimental_details}

\subsection{Additional Analysis Details}
\label{app:proof}

\begin{proof}
Let $Y_j=\mathbbm{1}\{L\geq j\}$ and choose score levels $0<q_{\mathrm{lo}}<\tau<q_{\mathrm{hi}}<1$. Draw $Z\sim\operatorname{Bernoulli}(\rho)$ and independent variables $Z_1,\ldots,Z_\gamma\sim\operatorname{Bernoulli}(\rho)$, all independently of $L$. Define
\[
q_j
=
q_{\mathrm{lo}}
+
(q_{\mathrm{hi}}-q_{\mathrm{lo}})Y_jZ
\]
\[
q'_j
=
q_{\mathrm{lo}}
+
(q_{\mathrm{hi}}-q_{\mathrm{lo}})Y_jZ_j.
\]
At any fixed position, both scores equal $q_{\mathrm{lo}}$ when $Y_j=0$; when $Y_j=1$, each equals $q_{\mathrm{hi}}$ with probability $\rho$ and $q_{\mathrm{lo}}$ otherwise. Their joint distributions with $Y_j$ are therefore identical. Their dependence across positions is different, however: the first $K$ scores in the first process exceed $\tau$ when $L\geq K$ and $Z=1$, whereas the second process additionally requires $Z_1=\cdots=Z_K=1$. These events have the two stated probabilities.
\end{proof}

The construction isolates the dependence that positionwise metrics cannot observe. The shared variable $Z$ makes high scores occur together in the first process, while the independent variables $Z_j$ disperse them across blocks in the second. Because the positionwise joint distributions remain identical, both processes receive the same value under any positionwise metric, including BCE and AUROC wherever defined. Nevertheless, a high-confidence prefix of length $K$ occurs $\rho^{1-K}$ times as often in the first process. Thus, even when every selected prefix is fully accepted, positionwise metrics do not determine how often the scheduler encounters an eligible prefix.

\subsection{Offline AUROC}
\label{app:offline_metrics}

Figure~\ref{fig:offline_auroc} compares the learned survival scores with raw confidence using positionwise AUROC. Marginal survival performs better across most positions, while conditional survival and raw confidence are closer.

\begin{figure*}[t]
\centering
\includegraphics[width=0.485\textwidth]
{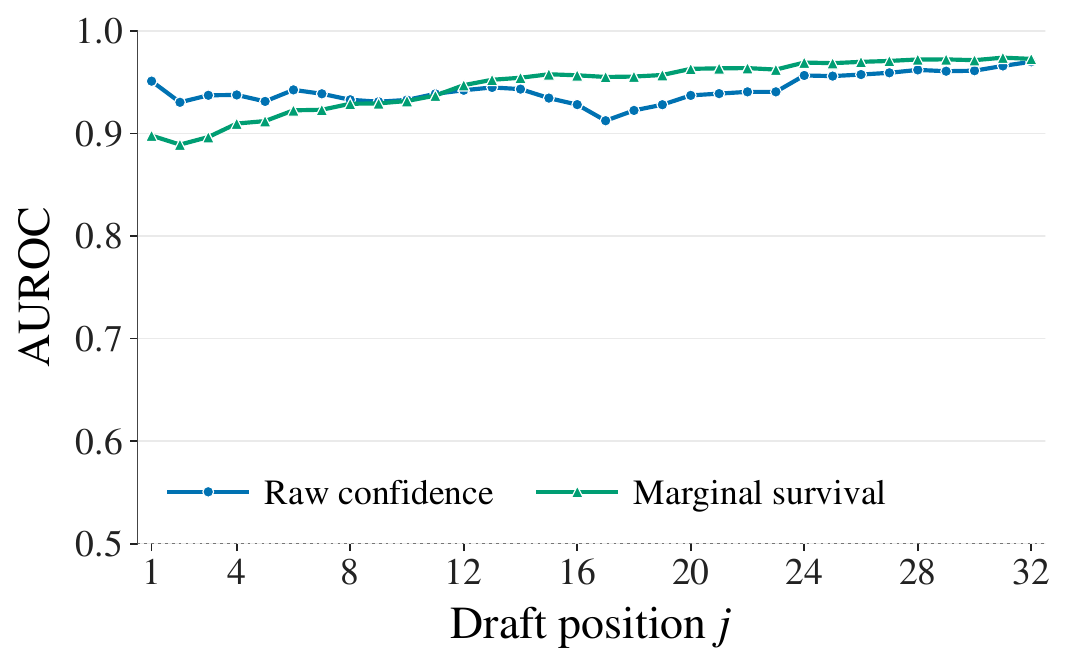}
\hfill
\includegraphics[width=0.485\textwidth]
{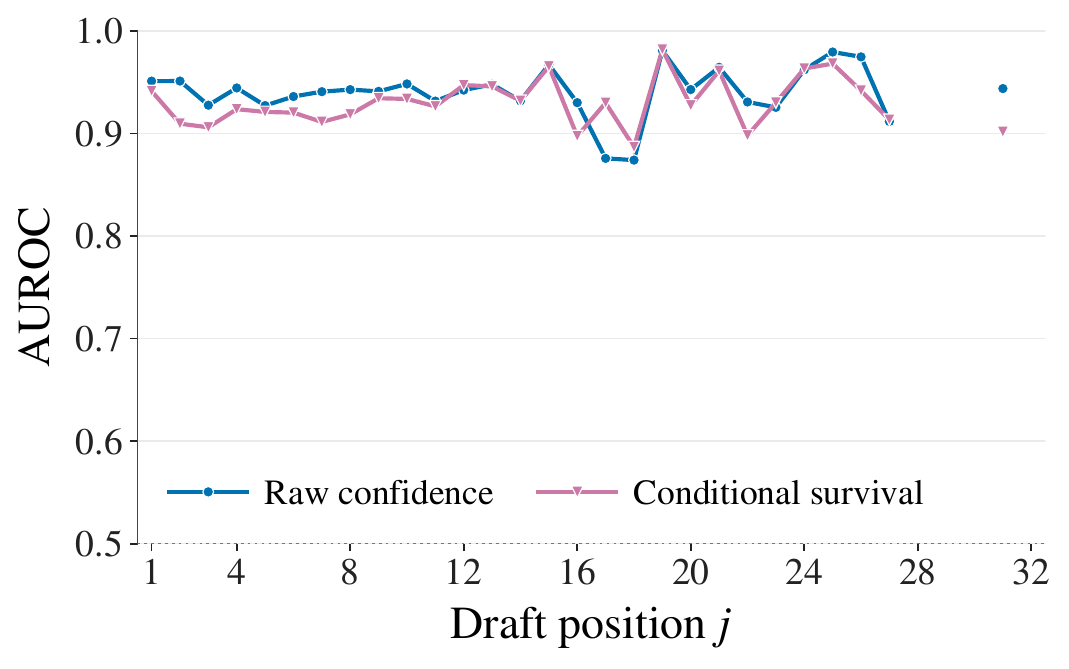}
\caption{
Positionwise AUROC on the held-out HumanEval folds with draft length $\gamma=32$. The conditional evaluation includes positions with at least ten positive and ten negative examples.
}
\label{fig:offline_auroc}
\end{figure*}

\subsection{Gate Ablation}
\label{app:gate_ablation}

Table~\ref{tab:gate_ablation} reports a fixed threshold ablation of the policy gates. We use settings fixed before inspecting the Qwen3 test results: $(\eta_b,\eta_c)=(0.90,0.93)$ for raw confidence and $(\eta_b,\eta_c)=(0.90,0.90)$ for conditional survival. Each variant removes one gate while leaving the remaining settings unchanged. We do not retune the variants because the goal is to identify the active constraint at each setting. We omit marginal survival because its local and prefix gates are equivalent.

\begin{table*}[t]
\centering
\caption{
Fixed threshold gate ablation on HumanEval at $\gamma=32$. Variants are not retuned. Calls and blocks are normalized by Strict SDD.
}
\label{tab:gate_ablation}
\small
\setlength{\tabcolsep}{6pt}
\begin{tabular}{@{}llcccc@{}}
\toprule
Signal & Policy & Pass@1 & Tokens/s & Relative calls & Relative blocks \\
\midrule
\multirow{3}{*}{Raw confidence}
& Full           & 0.8902 & 58.50 & 0.8170 & 1.135 \\
& Local only     & 0.8902 & 58.24 & 0.8146 & 1.135 \\
& Geometric only & 0.8598 & 60.82 & 0.7659 & 1.118 \\
\midrule
\multirow{3}{*}{\shortstack[l]{Conditional\\survival}}
& Full         & 0.8841 & 58.11 & 0.8175 & 1.137 \\
& Product only & 0.8841 & 58.11 & 0.8175 & 1.137 \\
& Local only   & 0.7805 & 58.96 & 0.7896 & 1.130 \\
\bottomrule
\end{tabular}
\end{table*}

For raw confidence, the local-only policy closely matches the full policy, whereas the geometric-only policy lowers pass@1 from $0.8902$ to $0.8598$. The local gate is therefore the main constraint at this setting. For conditional survival, the product-only policy is identical to the full policy, while the local-only policy reduces pass@1 to $0.7805$. Since $\eta_c=\eta_b=0.9$, passing the product gate also passes the local gate. This redundancy is specific to the chosen thresholds; other choices can make both gates active, as discussed in Section~\ref{subsec:confidence_signals}. 

\subsection{Periodic Skipping Baseline}
\label{app:periodic}

We evaluate a confidence-free baseline that attempts a skip of $K=6$ tokens once every $P$ rounds. All other rounds use strict decoding, and the same safety guards are checked before each scheduled skip.

\begin{table}[t]
\centering
\caption{
Periodic skipping on HumanEval at $\gamma=32$. Relative calls are normalized by Strict SDD, and agreement denotes strict token agreement. Shadow verification is excluded from calls and throughput.
}
\label{tab:periodic}
\small
\setlength{\tabcolsep}{3pt}
\begin{tabular}{@{}lcccc@{}}
\toprule
$P$ & Pass@1 & Tokens/s & Rel. calls & Agreement \\
\midrule
3 & 0.5854 & 63.42 & 0.7939 & 0.7417 \\
4 & 0.6768 & 60.50 & 0.8502 & 0.7542 \\
8 & 0.7256 & 56.80 & 0.9540 & 0.6989 \\
\bottomrule
\end{tabular}
\end{table}

Each periodic setting is dominated by a point from the raw-confidence sweep in both pass@1 and relative verifier calls.

\subsection{Transfer to MBPP+ and MATH}
\label{app:qwen25_transfer}

We report transfer results with DiffuCoder-7B-Instruct and Qwen2.5-72B-Instruct at $\gamma=16$ \cite{gong2025diffucoder,qwen2025qwen25technicalreport}. The conditional predictor was trained on HumanEval strict traces and evaluated on MBPP+ and MATH. The marginal predictor was trained on HumanEval and MATH strict traces and evaluated on MBPP+ only. Neither predictor used MBPP+ examples, and all policy thresholds were selected on HumanEval. For MBPP+, the HumanEval indicator is 1 and the MATH indicator is 0. For the conditional checkpoint, normalization maps the indicator inputs to zero to numerical precision on both HumanEval and MATH.

We evaluate all 378 MBPP+ problems
\cite{austin2021programsynthesislargelanguage,NEURIPS2023_43e9d647} and the first 300 problems, in their original order, from the MATH-500 test split \cite{NEURIPS_DATASETS,lightman2023letsverifystepstep}. All runs use greedy decoding, two diffusion steps, $K_{\min}=6$, and four NVIDIA A100 80 GB GPUs. Qwen2.5 is sharded over all four GPUs using \texttt{device\_map=auto} and a 50 GiB memory limit on GPU 0. DiffuCoder and the predictor also run on GPU 0.
For MBPP+, we use the MBPP+ v0.2.0 plus tests from EvalPlus and execute them directly with a 45-second timeout. The reported score is the fraction of problems that pass these tests. For MATH, we extract the last balanced boxed expression, or the final signed integer or decimal if no boxed expression is present. We normalize LaTeX wrappers and spacing, currency and percent notation, and numeric formatting before exact string comparison. We do not use symbolic equivalence checking. Throughput is not directly comparable with the Qwen3 experiments because the model pair and hardware configuration differ.

The MATH settings were fixed before evaluation. Raw confidence uses $S_{\max}=24$, whereas conditional survival uses $S_{\max}=32$, so these rows are not a controlled comparison.

\begin{table*}[t]
\centering
\caption{
Transfer results with DiffuCoder-7B-Instruct and Qwen2.5-72B-Instruct at $\gamma=16$. Score is the fraction of MBPP+ problems passing the plus tests and normalized answer exact match for the MATH subset. Calls are normalized by Strict SDD. Shadow verification is excluded from calls and throughput.
}
\label{tab:qwen25_transfer}
\small
\setlength{\tabcolsep}{5pt}
\begin{tabular}{@{}llccccc@{}}
\toprule
Dataset & Method & $S_{\max}$ & Score & Tokens/s & Relative calls & Strict agreement \\
\midrule
\multirow{4}{*}{MBPP+}
& Strict SDD           & -- & 0.7593 & 23.03 & 1.0000 & 1.0000 \\
& Raw confidence       & 32 & 0.7646 & 25.64 & 0.8318 & 0.7503 \\
& Marginal survival    & 32 & 0.7434 & 24.14 & 0.9106 & 0.9344 \\
& Conditional survival & 32 & 0.7513 & 23.35 & 0.9577 & 0.9544 \\
\midrule
\multirow{3}{*}{MATH subset}
& Strict SDD           & -- & 0.4733 & 15.85 & 1.0000 & 1.0000 \\
& Raw confidence       & 24 & 0.4767 & 16.53 & 0.9359 & 0.9656 \\
& Conditional survival & 32 & 0.4733 & 16.26 & 0.9636 & 0.9673 \\
\bottomrule
\end{tabular}
\end{table*}

On MBPP+, raw confidence gives the largest call reduction but the lowest strict agreement. Using 10,000 paired bootstrap resamples of the 378 tasks, the 95\% percentile intervals for the difference in score relative to Strict SDD are $[-0.0079,0.0185]$ for raw confidence, $[-0.0344,0.0026]$ for marginal survival, and $[-0.0212,0.0053]$ for conditional survival. All three include zero. On the MATH subset, all scores are within $0.01$ of Strict SDD, while relative verifier calls fall by $3.6\%$ to $6.4\%$.

\subsection{Baselines and Sweep Settings}
\label{app:baseline_details}

For lenience, let $z_j(v)$ be the target logit for token $v$ at position $j$. A drafted token $\hat{x}_j$ is accepted when
\begin{equation}
z_j(\hat{x}_j)-\max_v z_j(v) \geq \log \ell.
\label{eq:lenience_rule}
\end{equation}
The decoder accepts the longest prefix satisfying this condition and uses the target argmax at the first rejection. A fully accepted block uses the same target bonus token as Strict SDD. At $\ell=1$, we use the Strict SDD argmax comparison so that ties are handled identically.

Top-$k$ acceptance takes the longest draft prefix for which every drafted token is among the $k$ largest target logits at its position. At the first rejected position, it uses the target argmax. A fully accepted block uses the target bonus token. Both lenience and top-$k$ acceptance call the target model in every round. Longer accepted prefixes can still reduce verifier calls per generated token.

All main HumanEval confidence sweeps use $S_{\max}=64$ and, except in the $K_{\min}$ diagnostic, $K_{\min}=6$. A skip also requires $\widehat K>0$, so setting $K_{\min}=0$ removes only the minimum-length constraint. All three skip policies reject prefixes containing forbidden special tokens, unverified end-of-sequence tokens, repetition, invalid scores, or shortened final blocks. These guards are not applied to lenience or top-$k$ acceptance because they call the target model in every round.

\begin{table*}[t]
\centering
\caption{Hyperparameter grids used in the reported sweeps.}
\label{tab:sweep_grids}
\small
\setlength{\tabcolsep}{5pt}
\begin{tabular}{@{}lp{0.76\textwidth}@{}}
\toprule
Method & Values \\
\midrule
Lenience &
$\ell\in\{1.0,0.99,0.975,0.95,0.90,0.80,0.70,0.50,0.30\}$ \\
Top-$k$ &
$k\in\{2,3,5\}$; $k=1$ is Strict SDD \\
Raw confidence &
$\eta_b\in\{0.90,0.95\}$ and
$\eta_c\in\{0.995,0.99,0.98,0.97,0.95,0.93,0.90\}$, using every pairing;
additional pairs $(0.80,0.93)$, $(0.80,0.85)$, $(0.70,0.70)$, and
$(0.60,0.50)$ \\
Marginal survival &
$\eta_m\in\{0.85,0.90,0.93,0.95,0.97,0.98,0.99\}$ \\
Conditional survival &
$\eta_b=0.90$ and
$\eta_c\in\{0.99,0.98,0.97,0.95,0.93,0.90,0.85,0.80\}$ \\
\bottomrule
\end{tabular}
\end{table*}


\subsection{Predictor Details}
\label{app:predictor_details}

Raw confidence is the probability assigned to a selected token in the DiffuCoder forward pass in which that token is revealed. Because drafting uses two diffusion steps, positions may receive their scores at different steps.

The learned predictors are trained only on blocks from Strict SDD trajectories. Each predictor takes the 3,584-dimensional input to the LM head at the reveal step and 19 scalar features: reveal step, normalized reveal rank, reveal parallelism, a late-reveal indicator, top-1 probability and its logarithm, entropy and its negative, probability margin, top-5 probability mass, top-1 probability and negative entropy at diffusion step 0, their changes from diffusion step 0 to the reveal step, block position, generated length, segment position, and dataset indicators for HumanEval and MATH. The dataset indicators identify datasets rather than individual problems and are constant in the main HumanEval experiments. No verifier hidden state is used.

The MLP architecture is $3603\rightarrow256\rightarrow256\rightarrow256\rightarrow1$. The hidden layers use GELU, with dropout $0.1$ after the first two layers. Inputs are normalized using only the training folds. We train for 15 epochs with binary cross entropy, Adam, a learning rate of $10^{-3}$, and a batch size of 4,096 tokens. We select the checkpoint with the lowest validation BCE. The online checkpoints use training seed 0.

Five folds are formed by shuffling the 164 prompt IDs with seed 12345. In rotation $r$, fold $r$ is used for testing, fold $(r+1)\bmod 5$ for validation, and the other three folds for training. Validation is used only for checkpoint selection, not for choosing policy thresholds or reported operating points.

\subsection{Evaluation and Throughput}
\label{app:throughput_details}

We generate one greedy completion for each of the 164 HumanEval problems, using temperature zero and decoding seed 0. The prompt requests the full function in a single Python code block. We extract and compile the code, then run the official HumanEval tests in a subprocess with a 10-second timeout. Pass@1 is the fraction of problems that pass.

DiffuCoder and Qwen3 use BF16 weights, and Qwen3 logits are converted to FP32 when computing log probabilities. Generation uses a batch size of one on two NVIDIA A100 80 GB GPUs. DiffuCoder and the predictor run on one GPU; Qwen3 runs on the other without model parallelism.

Before timing, each process runs one draft block and one target forward pass. The GPUs are synchronized before and after each timed decode. Time includes drafting, prediction, target verification, and decoder control, but excludes model loading, warmup, prompt construction, HumanEval execution, and shadow verification. Each verifier call recomputes the full current sequence without a KV cache, including previously skipped tokens. This work is included in the reported time.

Each prompt and policy setting is timed once. We aggregate throughput as $\sum_i N_i/\sum_i T_i$, where $N_i$ is the number of generated tokens and $T_i$ is the decode time for prompt $i$. All Qwen3 HumanEval runs, including the periodic baseline and top-$k$ extension, use this timing procedure.

The 95\% intervals reported in Section~\ref{subsec:offline_online} use 10,000 paired bootstrap resamples of the HumanEval prompts with seed 0.

\subsection{Licenses}

DiffuCoder-7B-Instruct uses the Apple Model License for Research, Qwen3-32B uses Apache 2.0, and Qwen2.5-72B-Instruct uses the Qwen License Agreement dated September 19, 2024. HumanEval is MIT-licensed. MBPP+ is distributed under Apache 2.0, while the underlying MBPP data use CC BY 4.0. EvalPlus v0.3.1 uses Apache 2.0, with additional MIT terms for its HumanEval-derived evaluation code. MATH is MIT-licensed. The MATH-500 split comes from the MIT-licensed PRM800K repository; the HuggingFaceH4/MATH-500 mirror does not declare a separate license.

\end{document}